\documentclass[aps,pra,superscriptaddress,twocolumn,showpacs,titlepage=false]{revtex4-1}
\usepackage{graphicx}
\usepackage{amsmath}
\usepackage{amsmath}
\usepackage{mathrsfs}
\usepackage{dsfont}
\usepackage{amssymb}
\usepackage{subfigure}
\usepackage{color}
\usepackage{blkarray}
\usepackage[breaklinks=true,colorlinks=true,linkcolor=blue,urlcolor=blue,citecolor=blue]{hyperref}
\def\bra#1{{\left\langle #1 \right|}}
\def\ket#1{{\left| #1 \right\rangle}}

\definecolor{amber}{rgb}{1.0, 0.49, 0.0}
\usepackage[dvipsnames]{xcolor}

\usepackage{verbatim}
\usepackage{amsthm}
\usepackage{enumerate}
\usepackage{bbm}
\usepackage{hyperref}
\newtheorem{theorem}{Theorem}

\usepackage{verbatim}
\usepackage{graphicx}
\usepackage{amssymb}
\usepackage{amsmath}
\usepackage{amsthm}
\usepackage{enumerate}
\usepackage{bbm}
\usepackage{hyperref}

\newtheorem{lemma}{Lemma}
\newtheorem{corollary}{Corollary}
\newtheorem{remark}{Remark}

\newtheorem{proposition}{Proposition}

\newtheorem{definition}{Definition}

\newcommand*{\bbE}{\mathbb{E}}
\newcommand*{\bbI}{\mathbb{I}}

\newcommand*{\cE}{\mathcal{E}}

\newcommand*{\cL}{\mathcal{L}}
\newcommand*{\cM}{\mathcal{M}}
\newcommand*{\cN}{\mathcal{N}}

\newcommand*{\id}{\mathsf{id}}
\newcommand*{\proj}[1]{\ket{#1}\bra{#1}}
\newcommand*{\Tr}{\mathrm{Tr}}

\newcommand*{\myspan}{\mathrm{span}}

\newcommand{\beq}{\begin{equation}}
\newcommand{\eeq}{\end{equation}}
\providecommand{\av}[1]{\langle #1 \rangle}

\begin{document}

\title{Generic emergence of objectivity of observables in infinite dimensions}
\author{Paul A. Knott}
\email{Paul.Knott@nottingham.ac.uk}
\affiliation{Centre for the Mathematics and Theoretical Physics of Quantum Non-Equilibrium Systems (CQNE), School of Mathematical Sciences, University of Nottingham, University Park, Nottingham NG7 2RD, UK}
\author{Tommaso Tufarelli}
\affiliation{Centre for the Mathematics and Theoretical Physics of Quantum Non-Equilibrium Systems (CQNE), School of Mathematical Sciences, University of Nottingham, University Park, Nottingham NG7 2RD, UK}
\author{Marco Piani}
\affiliation{SUPA and Department of Physics, University of Strathclyde, Glasgow, G4 0NG, UK}
\author{Gerardo Adesso}
\affiliation{Centre for the Mathematics and Theoretical Physics of Quantum Non-Equilibrium Systems (CQNE), School of Mathematical Sciences, University of Nottingham, University Park, Nottingham NG7 2RD, UK}
\date{\today}

\begin{abstract}
Quantum Darwinism posits that information becomes objective whenever multiple observers indirectly probe a quantum system by each measuring a fraction of the environment. It was recently shown that objectivity of observables emerges generically from the mathematical structure of quantum mechanics, whenever the system of interest has finite dimensions and the number of environment fragments is large [F. G. S. L. Brand\~ao, M. Piani, and P. Horodecki,
Nature Commun. 6, 7908 (2015)]. Despite the importance of this result, it necessarily excludes many practical systems of interest that are infinite-dimensional, including harmonic oscillators. Extending the study of Quantum Darwinism to infinite dimensions is a nontrivial task: we tackle it here by using a modified diamond norm, suitable to quantify the distinguishability of channels in infinite dimensions. We prove two theorems that bound the emergence of objectivity, first for finite energy systems, and then for systems that can only be prepared in states with an exponential energy cut-off. We show that the latter class of states includes any bounded-energy subset of single-mode Gaussian states.
\end{abstract}
\maketitle

How does the objective classical world emerge from an underlying quantum substrate? The theories of decoherence and Quantum Darwinism (QD) can provide a rigorous framework to answer this question \cite{ZurekRMP,SchlosshauerRMP,ZurekQD,Ollivier2004,Ollivier2005,Blume2006,Horodecki2015}. Firstly, in decoherence we acknowledge that realistic systems are rarely isolated, but rather are coupled to an inaccessible environment. It can then be shown that, under suitable assumptions on the system-environment interaction, only states of a certain basis -- the pointer basis -- survive the system-environment interaction, while superpositions of these pointer states are suppressed. QD extends this by considering observers who interact indirectly with the system by having access to fragments of the environment, as illustrated in Fig.~\ref{artwork}: each observer measuring the system only has access to a single fragment of the environment (for example, we observe everyday objects by measuring a small fragment of the vast photon environment). This formalism can now be used to prove the emergence of objectivity: it has been shown, using various models \cite{Blume2008,Zwolak2009,Riedel2010,Riedel2012,Galve2015,Balaneskovic2015,Tuziemski2015,Tuziemski2015b,Tuziemski2016,Balaneskovic2016,Lampo2017,Pleasance2017}, that multiple observers measuring the same system will invariably agree on the outcomes of their measurements, and hence their outcomes can be regarded as objective.

\begin{figure}[t!]
\centering
\includegraphics[width=7.5cm]{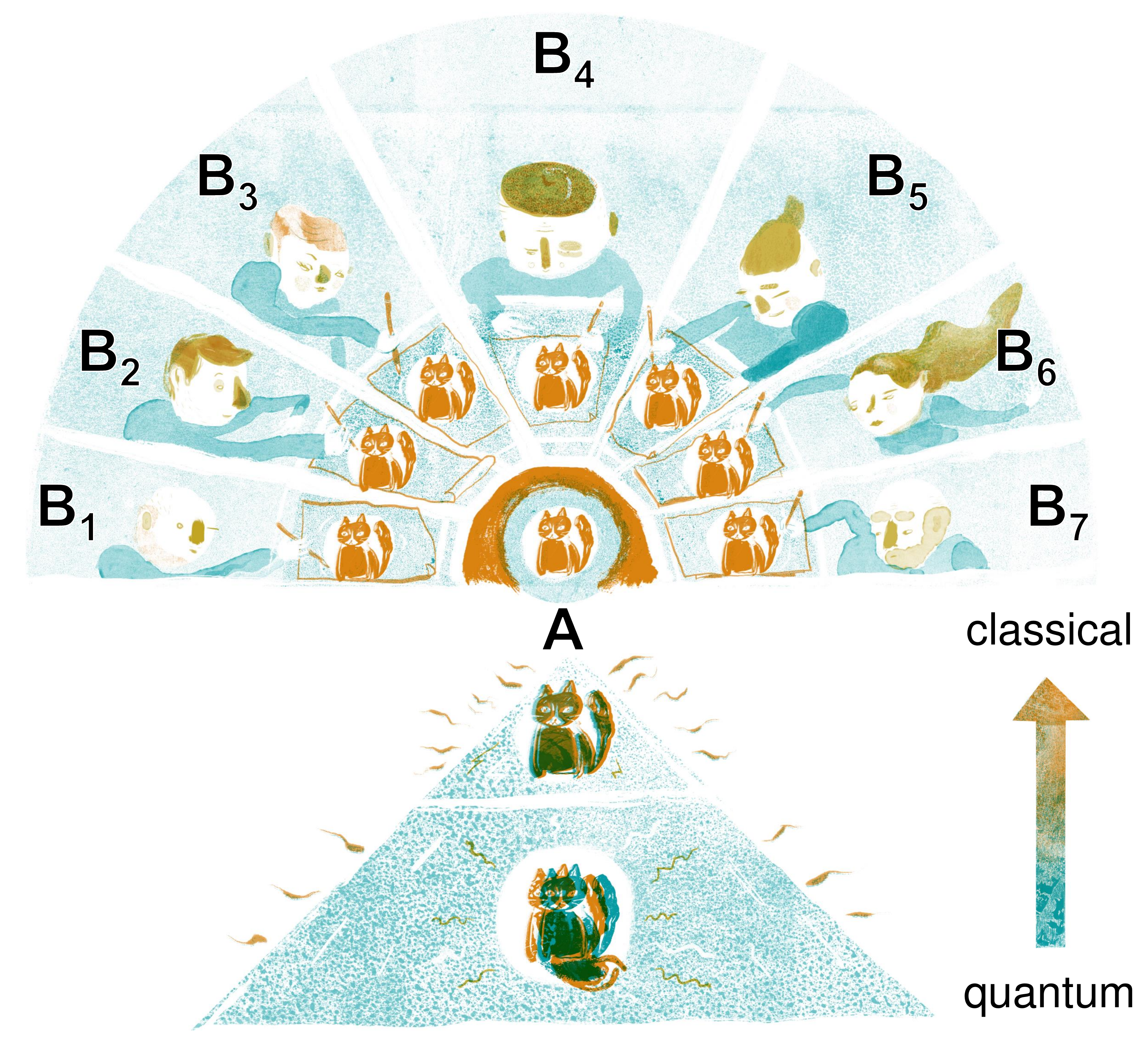}
\caption{In Quantum Darwinism, the objective classical reality (image of a cat, subsystem $A$ in the middle) emerges from an underlying quantum mechanical description  (bottom layer, illustrating superposition effects) through the observation of multiple environment fragments (subsystems $B_1 \ldots B_N$, depicted as painting artists). Here we show that the objectivity of observables is generic even when $A$ is an infinite dimensional quantum system, subject to suitable energy constraints.
\label{artwork}}
\end{figure}

To obtain a complete description of objectivity, we must also describe how the \emph{objectivity of observables} emerges -- namely, why do multiple observers measure the same observables? Indeed, if different observers had access to inequivalent observables, the very notion of objectivity of the measurement outcomes would be ill-defined. Objectivity of observables has been studied in numerous contexts, where it was shown that, when an observer only has access to a single fragment of the environment, the only information available is information about certain preferred observables, which correspond to the preferred basis \cite{ZurekQD,Ollivier2004,Ollivier2005,Blume2006,Horodecki2015,Blume2008,Zwolak2009,Riedel2010,Riedel2012,Galve2015,Balaneskovic2015,Tuziemski2015,Tuziemski2015b,Tuziemski2016,Balaneskovic2016,Lampo2017,Pleasance2017}. Until recently, the majority of research consisted of studying specific models; the question then remained of whether the emergence of objectivity is a generic feature, or only a model-specific one.

This changed with a result by Brand{\~a}o et al.~\cite{Brandao2015QD}, who showed that objectivity of observables is a generic phenomenon that emerges in a model-independent way from the basic mathematical structure of quantum mechanics, whenever the number of environment fragments gets large. Objectivity was there intended in the sense of information available and agreed upon about some general measurement performed indirectly on the system and described by a positive-operator-valued measure (POVM)~\cite{Nielsen:2011aa}. However, there was one caveat: the authors of \cite{Brandao2015QD} required the system of interest to live in a finite-dimensional Hilbert space. Despite the appeal of finite-dimensional results, they do not cover the physically very relevant case of continuous-variable systems.

In this paper we overcome this significant restriction to prove that infinite dimensional systems, under appropriate constraints, also exhibit objectivity of observables. Specifically, we show that objectivity emerges generically when considering either of the following physically motivated restrictions: (i) systems with finite mean energy --- which arguably include all realistic systems of interest; and (ii) systems with an exponential energy cut-off --- which include systems prepared in single-mode bosonic Gaussian states. In both cases, we prove exact bounds to show that, as the number of environment fragments grows large, objectivity of observables emerges. In \cite{Brandao2015QD}, the bound on objectivity depended only on the system dimension and the number of environment fragments; in contrast, our bounds provide non-trivial extensions that show an explicit dependence on the system's mean energy in the first case, and on the strength of the exponential cut-off in the second. Our results rely on a combination of mathematical techniques of potential independent interest,
and overall shed further light on the underlying structure of our physical reality.


As shown in Fig.~\ref{artwork}, the framework we study consists of a collection of (generally infinite-dimensional) subsystems. We can select any one of these as our system of interest and label it $A$; the rest of the subsystems, denoted $B_1$ to $B_N$, are then taken to be the $N$ different fragments of the environment of $A$. We then consider any completely positive trace-preserving (cptp) map, i.e., any quantum channel, $\Lambda : \mathcal{D}(A) \rightarrow \mathcal{D}(B_1 \otimes \ldots \otimes B_N)$ from the system to the environments. The next step is crucial to QD \cite{ZurekQD}: we assume that each observer who wishes to measure system $A$ can only do so by measuring one fragment of the environment. To model this, we define the channel $\Lambda_j := \Tr_{\backslash B_j} \circ \Lambda $ as the effective dynamics from $A$ to $B_j$, where $ \Tr_{\backslash B_j}$ indicates the partial trace over all fragments except $B_j$.

As anticipated we shall consider restrictions on the properties of system $A$, but notice that throughout this letter the fragments $B_1,...,B_N$ do not need to satisfy any constraint. Our main results are expressed through generalizations of the so-called diamond norm. The latter encapsulates the notion of best possible distinguishability between two different physical processes, as allowed by quantum mechanics, and does not take into account reasonable physical constraints, like the energy involved in the discrimination procedure. The generalizations we will consider will instead consider said restrictions, e.g., focusing on systems with bounded mean energy.  Specifically, we define (see also \cite{PLOS,Shirokov2017,Winter2017}):
\begin{definition} (Energy-constrained diamond norm) For a Hermiticity-preserving linear map $\Lambda: \mathcal{D}(A) \rightarrow \mathcal{D}(B)$, and a finite $\bar n > 0$, we define
\beq
\|\Lambda\|_{\Diamond\bar{n}}:=\sup_{\rho:\Tr(\rho \hat{n}_{A})\leq \bar n} \|\Lambda_{A}\otimes\id_{C}  [\rho]\|_1\,,\label{eq:cutnorm}
\eeq
where $C$ is an arbitrary system, $\hat{n}_{A}$ is the number operator only for subsystem ${A}$, and the supremum is calculated over all physical states $\rho$ of $AC$ such that the energy of the reduced state $\rho_{A}$ respects the indicated bound.
\end{definition}
In the above, we indicate with $\|X\|_1$ the trace norm of an operator. Given two states $\sigma_0$ and $\sigma_1$ of a system $S$, the trace norm of their difference, $\|\sigma_0-\sigma_1\|_1$, is directly linked to the ability to discriminate whether $S$ was prepared in either $\sigma_0$ or $\sigma_1$~\cite{Nielsen:2011aa}. The meaning of the energy-constrained diamond norm is then that of providing a measure of distinguishability of two evolutions $\Lambda_0$ and $\Lambda_1$, by considering in the above definition $\Lambda = \Lambda_0-\Lambda_1$ and an input state $\rho$ that is limited in energy. Based on this definition, we can prove the following:

\begin{theorem}\label{th:main1} Let $\Lambda : \mathcal{D}(A) \rightarrow \mathcal{D}(B)$ be a cptp map. Define $\Lambda_j := \Tr_{ \backslash B_j} \circ \Lambda$ as the effective dynamics from $\mathcal{D}(A)$ to $\mathcal{D}(B_j)$ and fix a number $0 < \delta <1$. Then there exists a POVM $\{M_k\}$ and a set $S \subseteq \{1,...,N\}$ with $|S|\geq(1-\delta)N$ such that, for all $j \in S$ (and for some finite energy $\bar{n}$), we have
\begin{align}
\label{eq:mainbound}
\|\Lambda_j - \cE_j  \|_{\Diamond\bar{n}}\leq&\frac{17}{\delta} \left(\frac{2}{7}\right)^{\tfrac{14}{17}}\left(\frac{\bar{n}^7}{N}\right)^{\tfrac{1}{17} }\approx {6.06\over \delta } \left(\frac{\bar{n}^7}{N}\right)^{\tfrac{1}{17} }\,,
\end{align}
where the measure-and-prepare channel $\cE_j$ is given by
\beq
\cE_j(X) := \sum_k \Tr (M_k X) \sigma_{j,k},
\eeq
for states $\sigma_{j,k} \in \mathcal{D}(B_j)$. Here both spaces of $A$ and $B$ can have infinite dimensions.
\end{theorem}

We prove this result (\textit{mean energy bound} for brevity) in Appendix~\ref{AppA}, focusing here on its interpretation. As explained in introducing the energy-constrained diamond norm, the LHS of Eq.~(\ref{eq:mainbound}) is a measure of how well the two channels $\Lambda_j$ and $\cE_j$ can be distinguished.
We see that, for a fixed $\delta$ and $\bar{n}$, the two channels become indistinguishable as the number of environment fragments $N$ grows large. This holds for a fraction $1-\delta$ of the environment fragments; $\delta$ can be set as close to $0$ as required, but this in turn affects the RHS and hence the minimum value of $N$ providing a meaningful bound. Putting this together, we see that, for at least a fraction $1-\delta $ of the sub-environments, any quantum channel from $A$ to $B_j$ becomes arbitrarily close to a measure-and-prepare (also known as entanglement-breaking~\cite{Horodecki_Shor_Ruskai_2003}) channel. Moreover, the measured observable is the same for all these environment fragments and hence is objective --- any observer who wishes to probe system $A$ by measuring a fragment $B_j\in S$ can \textit{at most} gain information about the single POVM $\{M_k\}$.


The RHS of Eq.~\eqref{eq:mainbound}, however, tends to zero very slowly with $N$, and therefore a huge number of environment fragments are needed to give an informative result. We illustrate this in Fig.~\ref{bound_scarso}, where we set $\bar{n} = 1$ and $\delta=0.01$ (i.e. we ask whether objectivity of observables holds for 99\% of the observers). We note that the unfavourable $N^{-1/{17}}$ scaling can be improved slightly: in order to obtain a neat analytical result for Theorem~\ref{th:main1} some approximations were taken, but a slightly tighter bound can be found numerically (see Appendix). For the case $\bar n=1$, for example, we numerically obtain a scaling closer to $N^{-1/{15}}$. Yet, even with this improved result the upper bound of $\|\Lambda_j - \cE_j  \|_{\Diamond\bar{n}}$ remains of the order of $0.1$ when $N=10^{60}.$

\begin{figure}[t!]
	\begin{center}
	\includegraphics[width=.8\linewidth]{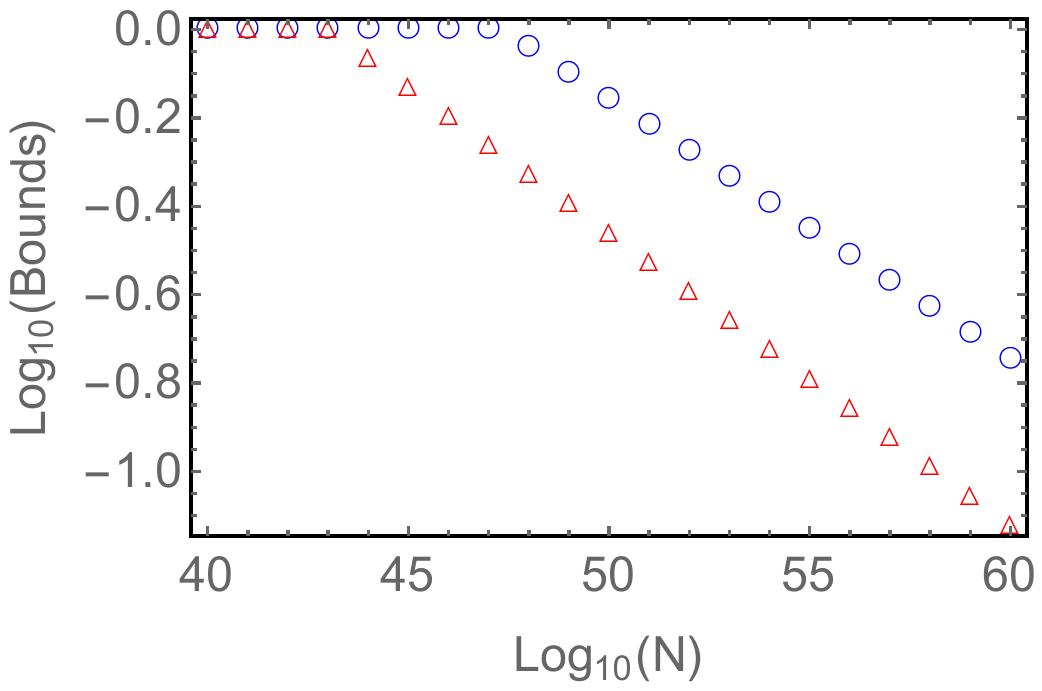}
	\end{center}
	\caption{Mean energy bound example. We plot the upper bound on $\|\Lambda_j - \cE_j  \|_{\Diamond\bar{n}}$ for $\bar n=1$ and $\delta=0.01$, as provided by Eq.~\eqref{eq:mainbound} (blue circles) and by a numerically optimized bound (red triangles, see Appendix). A power-law fit of the numerical bound provides $\|\Lambda_j - \cE_j  \|_{\Diamond\bar{n}}\leq\tfrac{1}{\delta}\beta N^{-\frac{1}{\alpha}}$ with $\beta\simeq 6.94,\alpha\simeq15.18$. \label{bound_scarso}}
\end{figure}

Bounded energy states arguably include all physically realizable states of a continuous-variable system (which we assume can be written in a single-mode Fock basis). Despite this, due to the unfavourable scaling of the bound with the number $N$ of environments, Theorem~\ref{th:main1} is of little practical use to show that objectivity of observables is generic in all such cases. 
To obtain a more informative bound, we can restrict further the class of states that are allowed in our diamond norm definition. In the following we will see that a more informative bound, admitting a vastly improved scaling with $N$, can be proved under such a restriction. The class of states we shall consider contains all the continuous variable density matrices with an\textit{ exponential energy cut-off} in subsystem $A$. Specifically, we restrict to all density matrices $\rho_{AC}$ such that
\begin{equation}\label{eq:cutoff}
\Tr\left[ \rho e^{\omega \hat{n}_{A}}\right] \leq \Omega\,,
\end{equation}
where $\Omega > 1$ and $\omega>0$ are given constants. We show below that the exponential cut-off states include meaningful subsets of single-mode bosonic Gaussian states, which play a focal role in continuous variable quantum information \cite{Adesso2014}. It is also trivial to show that any subset of states that can be written as a finite expansion in the Fock basis (up to some upper state $\ket{n_{\sf max}}$) belongs to this class. We can now define another variant of the diamond norm, relevant when only states obeying Eq.~\eqref{eq:cutoff} may be exploited to distinguish between channels.
\begin{definition} (Exponential cut-off diamond norm) For a Hermiticity-preserving linear map $\Lambda: \mathcal{D}(A) \rightarrow \mathcal{D}(B)$, and constants $\omega > 0$, $\Omega>1$, let
\beq \label{exp_cutoff}
\|\Lambda\|_{\Diamond\omega,\Omega}:=
\sup_{
	\substack{
		\Tr\left[ \rho e^{\omega \hat{n}_{A}} \right] \leq \Omega
	}
} \|\Lambda_{A} \otimes \id_C[\rho]\|_1 .
\eeq
where $C$ is an arbitrary ancillary system, $\hat{n}_{A}$ is the number operator only for subsystem $A$, and the supremum is calculated over all physical states $\rho$ of $AC$ such that the reduced state $\rho_{A'}$ respects the indicated bound.
\end{definition}
We can then prove the following:


\begin{theorem}\label{th:main2}
Let $\Lambda : \mathcal{D}(A) \rightarrow \mathcal{D}(B)$ be a cptp map. Define $\Lambda_j := \Tr_{ \backslash B_j} \circ \Lambda$ as the effective dynamics from $\mathcal{D}(A)$ to $\mathcal{D}(B_j)$ and fix a number $0<\delta <1$. Then there exists a POVM $\{M_k\}$ and a set $S \subseteq \{1,...,N\}$ with $|S| \geq (1-\delta)N$ such that, for all $j \in S$ (and for some finite $\omega>0$ and $\Omega>1$), we have
\begin{align}\label{expbound}
\|\Lambda_j - \cE_j  \|_{\diamond \omega,\Omega} &\leq
\frac{8}{\delta} \left(\frac{\gamma_1}{N}\right)^{1/3}\left[1+\frac14\big(\ln(\gamma_2 N)\big)^{4/3}\right]\,,
\end{align}
where the measure-and-prepare channel $\cE_j$ is given in Theorem~\ref{th:main1} and where
\begin{align}
\begin{split}
&\gamma_1 = \frac{2\tilde{d}^2s}{3\omega^4}\,, \quad \gamma_2=\frac{3 \tilde{d} \omega^4}{16s}\,, \\
&\tilde{n}=\displaystyle{1\over e^{\omega}-1}\,, \quad s= (\tilde{n}+1)\ln (\tilde{n}+1) - \tilde{n} \ln \tilde{n}\,.
\end{split}
\end{align}	
\end{theorem}

We prove this in Appendix~\ref{AppB}. Theorem~\ref{th:main2} can be interpreted similarly to Theorem~\ref{th:main1}: with increasing $N$, provided the available resources obey the exponential cut-off condition, any quantum channel from $A$ to a generic environment fragment $B_j\in S$ becomes arbitrarily close to the measure-and-prepare channel specified by the measurement of $\{M_k\}$ (again, the same measurement for all $j'$s). Let us remark, however, that the dominant scaling with $N$ is now $\propto(\ln{N})^{4/3}/N^{1/3} $, which for $N\to\infty$ converges to zero significantly faster than the right hand side of Eq.~\eqref{eq:mainbound}.
It is instructive also to compare our results to the results in \cite{Brandao2015QD}: their bound scales as $1/N^{1/3} $ and therefore approaches $0$ faster than our bounds in the limit $N \gg \infty$. However, with a straightforward modification of the proof of Theorem~\ref{th:main2} we can derive a new bound applicable to large (yet finite) dimensions, under the further assumption of an exponential energy cut-off as in Eq.~\eqref{eq:cutoff}. Such a bound can potentially be more informative than the one from \cite{Brandao2015QD} at finite (but still large) values of $N$, if the system satisfies Eq.~(\ref{eq:cutoff}) with suitable values of $\omega,\Omega$.

It is important at this point to provide relevant examples of quantum states that satisfy the exponential energy cut-off condition. Let us consider the case in which the reduced density matrix of $A$ is an arbitrary mixed single-mode Gaussian state $\rho_G$ \cite{Adesso2014}, specified by a displacement vector ${\bf d} = \sqrt 2 \{\Re (\alpha), \Im (\alpha)\}$ with $\alpha \in \mathds{C}$, and by a covariance matrix ${\bf V}= \mbox{diag}\{e^{2r}(2m+1), e^{-2r}(2m+1)\}$, where $m \geq 0$ is the mean number of thermal photons and we can fix without any loss of generality a squeezing parameter $r >0$. Note that we may assume the covariance matrix to be in diagonal form, since diagonalization can always be achieved via a phase rotation commuting with $\hat n$. By means of the Husimi function $Q(\beta)= \pi^{-1} \langle \beta|\rho|\beta\rangle$, where $\{|\beta\rangle\}$ (with $\beta \in \mathds{C}$) is the overcomplete basis of coherent states, we can evaluate the LHS~of Eq.~(\ref{eq:cutoff}) analytically, using the formula $\Tr\left[ \rho e^{\omega \hat{n}_{A}}\right] =e^{-\omega} \int_{\mathds{C}} d^2 \beta\  Q(\beta) e^{(1-e^{-\omega})|\beta|^2}$, which may be derived by anti-normally ordering the operator $e^{\omega \hat{n}_{A}}$ \cite{barnett2002methods}. We then find that a Gaussian state $\rho_G$ satisfies the exponential cut-off condition if and only if
\begin{equation}\label{eq:cutoffgaussian}
\av{e^{\omega \hat n}}=\frac{2 \exp\left[{\frac{2\Re(\alpha)^2}{\kappa_+^2}+\frac{2\Im(\alpha)^2}{\kappa_-^2}}\right]}{(e^\omega-1)\kappa_+ \kappa_-} \leq \Omega\,,
\end{equation}
with $\kappa_\pm=\sqrt{\coth \left(\frac{\omega }{2}\right)-(2 m+1) e^{\pm 2 r}}$. In Appendix~\ref{AppC}, we exploit this formula to show that any subset of Gaussian states with bounded energy, i.e. ${\cal G}_{\bar n}=\{\rho_G\vert \Tr[{\rho_G\hat n}]\leq \bar{n}\}$, obeys the desired cut-off condition whenever the parameters $\omega,\Omega$ satisfy
\begin{align}
\!\!\!\!\Omega&> 1/(1-\epsilon),\\
\!\!\!\!\omega&=\min\left\{ \frac{2\epsilon}{3/2+2\bar n(2+\bar n)}\, ,\,\frac{1-\epsilon}{\bar n}\ln\big((1-\epsilon)\Omega\big) \right\},
\end{align}
where $0<\epsilon<1$ is arbitrary parameter that can be tuned to optimize the resulting exponential cut-off bound. Note also that, once the relevant parameters have been fixed according to the above discussion, the entire convex hull of ${\cal G}_{\bar n}$ will also satisfy the exponential cut-off condition. For example, suppose one would like to distinguish between $\Lambda_j$ and $\mathcal{E}_j$, only being able to prepare mixtures of Gaussians with $\av{\hat n_{A}}\leq 1$. Then the results of Theorem~\ref{th:main2} (\textit{exponential cut-off bound} for brevity) would apply, giving us much tighter constraints on the emergence of objectivity as compared to what Theorem~\ref{th:main1} would tell us under the same hypothesis. Furthermore, also in this case we can consider a numerical optimization yielding an improved upper bound for the RHS of Eq.~\eqref{expbound} --- see Appendix. This is shown in Fig.~\ref{bound_mejo}, where fixing $\delta=0.01$ and $\bar n=1$, as before, we obtain $ \|\Lambda_j - \cE_j  \|_{\diamond \omega,\Omega}<0.5\times10^{-3}$ already for $N=10^{29}$. \\

\begin{figure}[t!]
	\begin{center}
		\includegraphics[width=.8\linewidth]{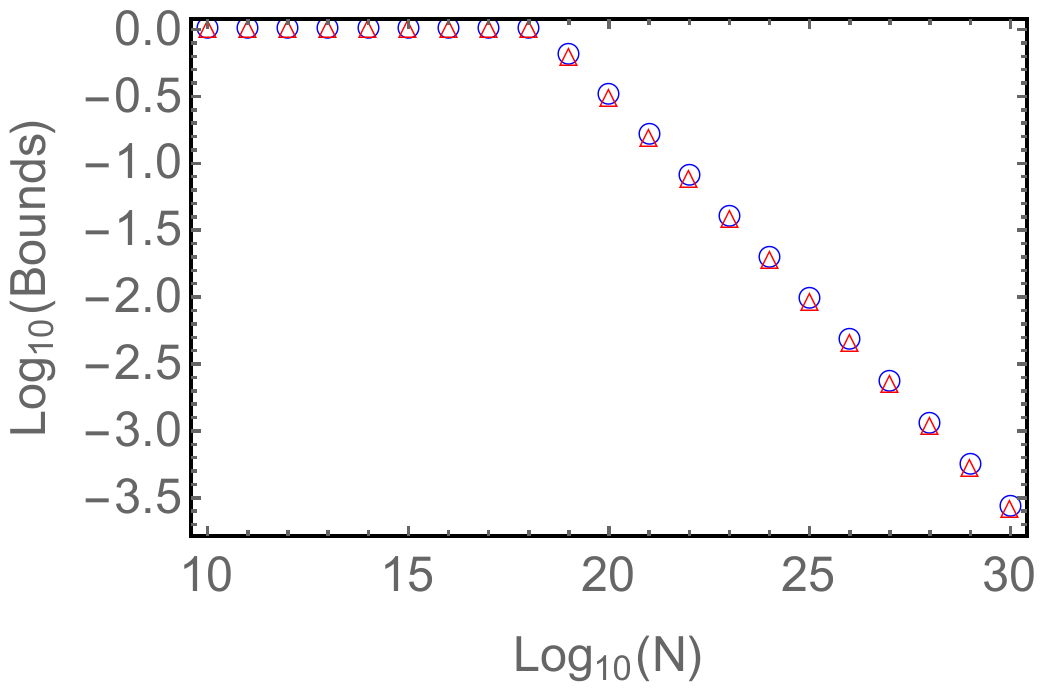}
	\end{center}
	\caption{Exponential cut-off bound example. We take $\delta=0.01$ as in Fig.~\ref{bound_scarso}. The parameters $\omega,\Omega$ are optimized for each $N$ to provide the best possible bound attainable through our methods, assuming a set of Gaussian resource states ${\cal G}_{\bar n}$ with ${\bar n}=1$ (these would be the input states that can be used to discriminate between channels). Blue circles indicate the RHS of Eq.~\eqref{expbound}, while red triangles refer to a numerically optimized bound (see Appendix). A power-law fit of the latter yields $\|\Lambda_j - \cE_j  \|_{\diamond \omega,\Omega}\leq\tfrac{1}{\delta}\beta N^{-\frac{1}{\alpha}}$ with $\beta\simeq 5839,\alpha\simeq3.20$. \label{bound_mejo}}
\end{figure}

It is remarkable that Theorems~\ref{th:main1} and \ref{th:main2} hold for \emph{any} channel $\Lambda$. These results, together with \cite{Brandao2015QD}, show that objectivity of observables is built into the basic mathematical structure of quantum mechanics. Despite this, our analysis suggests that objectivity of observables may emerge extremely slowly in the absence of further restrictions on the system's properties and/or the form of its interaction with the environment: generally a vast number of environment fragments is needed to have truly informative bounds. This might simply be due to our energy-based bounds not being tight, and indeed future work should go towards improving such bounds. There are nonetheless a number of other directions for future studies.

Firstly, following our approach, further restrictions can be placed on the set of states defining the generalized diamond norms. A bound specifically for Gaussian states would be an interesting next step, as would be a bound that combines \cite{Brandao2015QD} with our results to consider, for example, finite dimensional systems with a fixed energy. It would also be instructive to consider which states are responsible for allowing the channels in Theorems~\ref{th:main1} and \ref{th:main2} to be easily distinguished. Do we expect these states to be realistic, experimentally producible states? Alternatively, restrictions can be placed on the measurements available to distinguish between the channels; considering only coarse-grained measurements would bring us closer to real-world scenarios where objectivity emerges.

A different line of attack would be to place restrictions on the allowed channels. One of the big questions that quantum Darwinism addresses is: why does our macroscopic every-day world appear objective and classical, despite being constructed of quantum mechanical particles? The properties of these particles, at first sight, seem to be far from objective. With this in mind we may ask: What general properties are shared by physically meaningful channels? If all relevant interactions between system and environment are ultimately due to, for example, a combination of one- and two-body Hamiltonians, the resulting channels may display some nontrivial structure. Less ambitiously, one may place further intuitive restrictions on the channels, such as the conservation of a global number operator of the form $\hat n_A+\hat n_{B_1}+...+\hat{n}_{B_N}$ or other symmetry constraints.
By investigating these generalizations of our results, it may be possible to further clarify why our macroscopic world is ``classical'', without resorting to the task of constructing macroscopic models and simulations, the like of which will only be possible with macroscopic-scale quantum computers. \\

\begin{acknowledgments}
We thank Joseph Hollis for the artwork featured in Figure~1.
We acknowledge discussions with F. G. S. L. Brand\~ao, L. Lami, and A. Winter.
 This work was supported by the Foundational Questions Institute (fqxi.org) under the Physics of the Observer Programme (Grant No.~FQXi-RFP-1601), the European Research Council (ERC) under the Starting Grant GQCOP (Grant No.~637352), and the European Union’s Horizon 2020 Research and Innovation Programme under the Marie Skłodowska-Curie Action OPERACQC (Grant Agreement No. 661338). T.T.
 acknowledges financial support from the University of Nottingham via a Nottingham Research Fellowship.
\end{acknowledgments}

\bibliographystyle{apsrevfixedwithtitles}
\bibliography{QDinf}


\onecolumngrid
\appendix
\section{Proving theorem one: finite energy bound \label{AppA}}

\begin{definition} (Diamond norm for superoperators) For a Hermiticity-preserving superoperator $\Lambda:\cL(A)\rightarrow\cL(B)$, with $\cL(A)$ and $\cL(B)$ the spaces of linear operators on $A$ and $B$, respectively, we define \cite{watrous2004notes}
\beq\label{diamond_super}
\|\Lambda\|_{\Diamond} := \sup_{X_{AC}\neq 0} \{ \|(\Lambda_A \otimes \id_C)[X_{AC}]\|_1 / \| X_{AC} \|_1 \}.
\eeq
where the supremum is over all Hermitian operators $X_{AC}\neq 0$ that pertain to the input space of $\Lambda$ and to an arbitrary ancillary system $C$.
\end{definition}

\begin{definition} (Energy-constrained diamond norm) For a Hermiticity-preserving linear map $\Lambda $ on $A$ we define (see also \cite{PLOS,Shirokov2017,Winter2017})
\beq
\|\Lambda\|_{\Diamond\bar{n}}:=\sup_{\rho:\Tr(\rho \hat{n}_{A})\leq \bar n} \| \Lambda_{A}\otimes\id_C [\rho_{AC}]\|_1 .
\eeq
where the supremum is over density matrices $\rho $ pertaining to the input $A$ and an arbitrary ancilla $C$, and finite $\bar{n}>0$. Here $\hat{n}_{A}$ is the number operator only for subsystem $A$.
\end{definition}


\begin{remark}\label{remark1}
For any state $\rho_{AC}$ that satisfies $\Tr(\rho \hat{n}_{A})\leq \bar n$ there is a pure state $\psi=\ket{\psi}\bra{\psi}$ on an extended space ${ACC'}$  that also satisfies $\Tr(\psi \hat{n}_{A})\leq \bar n$ and such that $\|\Lambda_{A}\otimes \id_C [\rho]\|_1\leq \|\Lambda_{A}\otimes \id_{CC'} [\psi]\|_1$, so we can restrict our attention to pure states. Thanks to the Schmidt decomposition, we know that local support of $\psi$ on $CC'$ is isomorphic to $A$. Thus,  one can always imagine that, effectively, for the optimal choice of input, $CC'\simeq A'$, with $A'$ isomorphic to $A$. 
\end{remark}

\begin{theorem}
$\|\Lambda\|_{\Diamond\bar{n}}$ is a norm for any $\bar{n}>0$.
\end{theorem}
\begin{proof}
The property $\|\alpha \Lambda\|_{\Diamond\bar{n}}= |\alpha| \|\Lambda\|_{\Diamond\bar{n}}$ is directly inherited from the 1-norm; the same holds for the triangle inequality. To prove that $\|\Lambda\|_{\Diamond\bar{n}}>0$ for any non-vanishing $\Lambda$, as long as $\bar{n}>0$, it is enough to consider that, if $\Lambda\geq 0$, then there is at least one state $\proj{\psi}$ such that $\Lambda[\proj{\psi}]\neq 0$; either such a $\ket{\psi}=\ket{0}$ (and we are done), or we can consider the convex combination $(1-p)\proj{0}+p\proj{\psi}$ for small enough $p>0$.
\end{proof}

\begin{lemma}(Fock cutoff lemma from \cite{killoran2011strong}.)
For any normalized state $\rho$ it holds
\[
\Tr(\Pi_{d} \rho) \geq 1 -  \frac{\langle \hat{n} \rangle_\rho - \langle \hat{n} \rangle_{\rho_{{T}}}}{d}
\]
where $\Pi_{d} =\sum_{i=0}^{d-1} \proj{i}$, $\rho_T = \Pi_{d} \rho \Pi_{d}$ and we adopt the slight abuse of notation $ \langle \hat{n} \rangle_{\rho_{{T}}}=\Tr(\hat n\rho_T)$.
\end{lemma}
\begin{lemma}(Gentle measurement lemma, taken from \cite{winter1999coding})
Consider a density operator $\rho$ and an effect (a POVM element) $0\leq M\leq \openone$. Suppose
$\Tr(M\rho)\geq 1-\epsilon$; then
\[
\frac{1}{2}\|\rho-\rho'\|_1 \leq \sqrt{\epsilon},
\]
where $\rho'$ is the post-selected state $\rho' = \sqrt{M}\rho\sqrt{M}/\Tr(M\rho)$.
\end{lemma}

Putting the two lemmas together, we have the following
\begin{proposition}
	\label{prop:focktruncation}
Let $\rho=\rho_{AC}$ be a bipartite state, and define the truncated states $\rho_{{T}} := (\Pi_{d}\otimes\openone)\rho (\Pi_{d}\otimes\openone)$, $\rho_{{T\cN}}:=(\Pi_{d}\otimes\openone)\rho (\Pi_{d}\otimes\openone) / \Tr(\Pi_{d} \rho_{A})$. Then
\[
\frac{1}{2}\|\rho-\rho_{{T\cN}}\|_1 \leq \sqrt{\frac{\langle \hat{n}_{A} \rangle_\rho - \langle \hat{n}_{A} \rangle_{\rho_{_{T}}}}{d} }.
\]
\end{proposition}

\begin{definition}
	\label{def:truncchoi}
	Let $\Phi_{AA'} = \ket{\Phi}\bra{\Phi}$,  $\ket{\Phi} = (d)^{-1/2} \sum_{k =0}^{d-1} \ket{k,k} $ be a $d$-dimensional maximally entangled state up to local Fock state $\ket{d-1}$. For any cptp map $\Lambda : \mathcal{D}(A) \rightarrow \mathcal{D}(B)$ we define the truncated Choi-Jamio{\l}kowski operator of $\Lambda$ as
	\begin{equation}
	\label{eq:truncatedCJ}
	J_T  :=  \id_{A}\otimes \Lambda_{A'} (\Phi_{AA'}).
	\end{equation}
\end{definition}

\begin{remark}
Notice that the choice of labels as well as of ordering of the systems is unimportant, as long as used appropriately and consistently in any given circumstance. In particular, notice that the two systems $A$ and $A'$ in the definition \ref{def:truncchoi} of the truncated Choi-Jamio{\l}koski operator are isomorphic, so that $\Lambda: \mathcal{D}(A) \rightarrow \mathcal{D}(B)$ can be applied to $A'$ so to have that its truncated Choi-Jamio{\l}koski operator pertains to $AB$, that is, $J_T=J_T^{AB}$. Notice also that the state $\ket{\Phi}$ is $AA'$-symmetric.
\end{remark}

\begin{lemma}
\label{lem:trunc_lem_2}
(Truncated version of Lemma 2 in the Supplemental Material of \cite{Brandao2015QD}.)
For two cptp maps $\Lambda_0$ and $\Lambda_1$ it holds that
\beq
\|\Lambda_0 - \Lambda_1 \|_{\Diamond\bar{n}} \leq d \| J_T(\Lambda_0) - J_T(\Lambda_1) \|_1 + 2 \|\Lambda_0 - \Lambda_1\|_\Diamond \sqrt{\frac{\langle \hat{n}_{A} \rangle_\rho - \langle \hat{n}_{A} \rangle_{\rho_{_{T}}}}{d} }. 
\eeq
\end{lemma}
\begin{remark}
It should be clear that the truncated Choi-Jamiolkowski state \eqref{eq:truncatedCJ} is not in general in one-to-one correspondence with maps, the reason being that two maps may differ in their behaviour only on high-Fock-number states, so that it might be $J_T(\Lambda_0) = J_T(\Lambda_1)$ even if---in the extreme, but still possible case---the two maps are perfectly distinguishable. Notice that, on the contrary, $\|\cdot\|_{\Diamond\bar{n}}$ is a norm, so the left-hand side is non-zero as soon as $\Lambda_0\neq \Lambda_1$.
\end{remark}

\begin{proof}
Let $\rho=\rho_{AC}$ be the state optimal for the energy-constrained diamond norm \footnote{Such a state may not exist, because the set of states satisfying $\Tr(\rho \hat{n}_A)\leq \bar n$ is not compact; nonetheless, we can consider a state $\rho=\rho_{AC}$ that is optimal for the energy-constrained diamond norm within $\epsilon$, for an arbitrary $\epsilon>0$.}
\beq\label{eq:opt_cutdiamond}
\|\Lambda_0 - \Lambda_1 \|_{\Diamond\bar{n}} = \|(\Lambda_0 - \Lambda_1)_{A}\otimes \id_C [\rho]\|_1.
\eeq
We find
\beq
\begin{aligned}
\|(\Lambda_0 - \Lambda_1)_{A}\otimes\id_C [\rho]\|_1 - \|(\Lambda_0 - \Lambda_1)_{A}\otimes\id_C [\rho_{_{T\cN}}]\|_1
&\leq \|(\Lambda_0 - \Lambda_1)\otimes\id_C [\rho - \rho_{_{T\cN}}] \|_1 \\
&\leq \|\Lambda_0 - \Lambda_1\|_\Diamond \|\rho - \rho_{_{T\cN}}\|_1.
\end{aligned}
\label{eq:bdd_TJC}
\eeq
The first inequality follows from the reverse triangle inequality, and the second inequality from the definition of the diamond norm for superoperators in Eq.~(\ref{diamond_super}). We thus have
\beq
\begin{aligned}
\| \Lambda_0 - \Lambda_1\|_{\Diamond\bar{n}}
\leq
\|(\Lambda_0 - \Lambda_1)_{A}\otimes \id_{A'} [\rho_{_{T\cN}}]\|_1 + 2 \|\Lambda_0 - \Lambda_1\|_\Diamond \sqrt{\frac{\langle \hat{n}_{A} \rangle_\rho - \langle \hat{n}_{A} \rangle_{\rho_{_{T}}}}{d} },
\end{aligned}
\eeq
having used Proposition \ref{prop:focktruncation}. 
Notice that $\rho_{_{T\cN}}$ has a $d$-dimensional local support on $A$, and a purification $\ket{\psi}_{AA'A''}$ of it on $A''$ (with $A'A''$ taken to be one party, that is then effectively $d$-dimensional) used as input can only provide better distinguishability. For such a pure state we can use Lemma 2 from \cite{Brandao2015QD}, giving
\beq
\begin{aligned}
\| \Lambda_0 - \Lambda_1\|_{\Diamond\bar{n}}
\leq
d \| J_T(\Lambda_0) - J_T(\Lambda_1) \|_1 + 2 \|\Lambda_0 - \Lambda_1\|_\Diamond \sqrt{\frac{\langle \hat{n}_{A} \rangle_\rho - \langle \hat{n}_{A} \rangle_{\rho_{_{T}}}}{d} }.
\end{aligned}
\eeq

\end{proof}

\begin{remark}
We have
\beq
2 \|\Lambda_0 - \Lambda_1\|_\Diamond \sqrt{\frac{\langle \hat{n}_{A} \rangle_\rho - \langle \hat{n}_{A} \rangle_{\rho_{_{T}}}}{d} } \leq 4\sqrt{\bar{n} \over d},
\eeq
which holds because $\|\Lambda_0 - \Lambda_1\|_\Diamond \leq 2$, $\langle \hat{n}_{A} \rangle_\rho \leq \bar{n}$ and $\langle \hat{n}_{A} \rangle_{\rho_{_{T}}} \geq 0$.
\end{remark}

\begin{remark}
Lemma \ref{lem:trunc_lem_2} reduces to Lemma 2 in \cite{Brandao2015QD} for finite dimensions because, by taking the truncation to be the whole space, we have $\langle \hat{n}_{A} \rangle_\rho = \langle \hat{n}_{A} \rangle_{\rho_{_{T}}}$ and $J_T=J$.
\end{remark}

\begin{theorem}\label{th:main1_apx} Let $\Lambda : \mathcal{D}(A) \rightarrow \mathcal{D}(B)$, with $B=B_1B_2\ldots B_N$, be a cptp map. Define $\Lambda_j := \Tr_{ \backslash B_j} \circ \Lambda$ as the effective dynamics from $\mathcal{D}(A)$ to $\mathcal{D}(B_j)$ and fix a number $1 > \delta > 0$. Then there exists a POVM $\{M_l\}$ and a set $S \subseteq \{1,...,N\}$ with $|S| \geq (1-\delta)N$ such that, for all $j \in S$ (and for some finite energy $\bar{n}$, truncation $d$, and $k \leq N$), we have
\beq
\|\Lambda_j - \cE_j  \|_{\Diamond\bar{n}} \leq {1 \over \delta} \zeta
\eeq
where
\beq
\zeta = \sqrt{2 \ln(2) d^6 \log d \over k}+4\sqrt{\bar{n} \over d} + {2k \over N}
\eeq
and where the measure-and-prepare channel $\cE_j$ is given by
\beq
\cE_j(X) := \sum_l \Tr (M_l X) \sigma_{j,l},
\eeq
for states $\sigma_{j,l} \in \mathcal{D}(B_j)$. Here both spaces $A$ and $B$ can have infinite dimensions.
\end{theorem}

\noindent Before giving the proof we need a corollary of the proof of Theorem 1 in \cite{Brandao2015QD} (see Eq.~(16) in the Supplementary material of \cite{Brandao2015QD}).

\begin{corollary}\label{cor:brandao1a}

For a cptp map $\Lambda : \mathcal{D}(A) \rightarrow \mathcal{D}(B)$, with $B=B_1B_2\ldots B_N$, we define the truncated Choi-Jamio\l kowski state of $\Lambda$ as $\rho_{AB_1...B_N}  :=  \id_{A}\otimes\Lambda_{A'} (\Phi_{AA'})$, with the $d$-dimensional maximally entangled state $\Phi_{AA'} = d^{-1} \sum_{k,k'=0}^{d-1} \ket{k,k}\bra{k',k'}$.  The quantum-classical channels ${\cal M}_1,...,{\cal M}_N$ are defined as ${\cal M}_i(X) := \sum_l \Tr(N_{i,l}X) \ket{l}\bra{l}$, for given POVMs $\{N_{i,l}\}_l$. Fixing an integer $m$, there exists another integer $q\leq m$ and a set of indices $J:=(j_1,...,j_{q-1})$, such that
\beq
\bbE_{j \notin J} \max_{{\cal M}_j} \| \id_{A} \otimes {\cal M}_j \left( \rho_{AB_j} - \bbE_z \rho_{A}^z \otimes \rho_{B_j}^z \right) \|_1 \leq \sqrt{2 \ln(2) \log(d) \over m}.
\eeq
Here the $j $'s are sampled from the uniform distribution over $\{0,1,...,N\}$ without replacement and $\rho_{AB_j}^z$ is the post-measurement state on $AB_j$ conditioned on obtaining $z=\{l_{j_1},...,l_{j_{q-1}}\}$ when measuring ${\cal M}_{j_1},...,{\cal M}_{j_{q-1}}$ in the subsystems $B_{j_1},...,B_{j_{q-1}}$ of $\rho$. The associated probability reads
\beq
p(z)=\Tr\left({\id_{A} \otimes \Lambda_{A'} (\Phi_{AA'})\,N_{j_1,l_{j_1}}\otimes...\otimes N_{j_{q-1},l_{j_{q-1}}}}\right)
\eeq
(see \cite{Brandao2015QD} for a more thorough description of these results).
\end{corollary}

\noindent We can now prove Theorem \ref{th:main1_apx}.
\begin{proof}

All the states involved in Corollary \ref{cor:brandao1a} have support on a $d-$dimensional subspace of system $A$ and system $A'$; we can exploit Lemma 1 in \cite{Brandao2015QD} to obtain
\begin{align}
\|\rho_{AB_j} &- \bbE_z \rho_{A}^z \otimes \rho_{B_j}^z\|_1	 \\
&\leq d^2 \max_{{\cal M}_j} \| \id_{A} \otimes {\cal M}_j \left( \rho_{AB_j} - \bbE_z \rho_{A}^z \otimes \rho_{B_j}^z \right) \|_1. \notag
\end{align}
Therefore, by combining this with Corollary \ref{cor:brandao1a}, we get

\begin{align}\label{18-l}
\bbE_{j \notin J}\|\rho_{AB_j} &- \bbE_z \rho_{A}^z \otimes \rho_{B_j}^z\|_1	 \\
&\leq d^2 \bbE_{j \notin J} \max_{{\cal M}_j} \| \id_{A} \otimes {\cal M}_j \left( \rho_{AB_j} - \bbE_z \rho_{A}^z \otimes \rho_{B_j}^z \right) \|_1 \\
&= \sqrt{2 \ln(2) d^4 \log(d) \over m} .
\end{align}

Using Lemma \ref{lem:trunc_lem_2} we can bound the \emph{energy-constrained diamond norm} distance of two maps by the distance of their truncated Choi-Jamio\l kowski states, to find
\beq\label{20}
\| \Lambda_j - \cE_j  \|_{\Diamond\bar{n}} \leq d \| \rho_{AB_j} - \bbE_z \rho_{A}^z \otimes \rho_{B_j}^z \|_1 + 4\sqrt{\bar{n} \over d},
\eeq
where the action of  $\cE_j$ is explicitly that of measure-and-prepare channel on the lowest-energy $d$-dimensional subspace, as given by
\beq
\cE_j(X) := d\, \bbE_z \Tr ((\rho_A^z)^T X) \rho_{B_j}^z,
\eeq
for $X$ an Hermitian operator with support on $\myspan\{\ket{0},\ket{1},\ldots,\ket{d-1}\}$.
We can complete the action of $\cE_j$ as we wish on the orthogonal space, so to make it overall a measure-and-prepare channel on the whole space; e.g., we can define $\cE_j(X)=\Tr(X)\sigma$, for $\sigma$ a fixed normalized state of $B_j$, for any Hermitian operator $X$ with support on $\myspan\{\ket{d},\ket{d+1},\ldots\}$, then defining the general action of $\cE_j$ on an arbitrary Hermitian operator $X$ as $\cE_j(X)=\cE_j(\Pi_dX\Pi_d)+cE_j((\openone-\Pi_d)X(\openone-\Pi_d)).$
Note that the POVM $\{ d p(z) (\rho_A^z)^T \}$ is independent of $j$. It can easily be confirmed that the truncated Choi-Jamio\l kowski state of $\cE_j$ is indeed $\bbE_z \rho_A^z \otimes \rho_{B_j}^z$. Note also that the truncated Choi-Jamio\l kowski state of $\Lambda_j$ is $\rho_{AB_j}$, by definition. We then have
\begin{align}
\bbE_{j \notin J} \| \Lambda_j - \cE_j  \|_{\Diamond\bar{n}} &\leq d \bbE_{j \notin J} \| \rho_{AB_j} - \bbE_z \rho_A^z \otimes \rho_{B_j}^z \|_1 + 4\sqrt{\bar{n} \over d} \\
&\leq \sqrt{2 \ln(2) d^6 \log d \over m} + 4\sqrt{\bar{n} \over d}.
\end{align}
Then
\begin{align}
\bbE_j \| \Lambda_j - \cE_j  \|_{\Diamond\bar{n}} &={\mathbb P}(j \notin J) \bbE_{j \notin J} \| \Lambda_j - \cE_j  \|_{\Diamond\bar{n}} + {\mathbb P}(j \in J) \bbE_{j \in J} \| \Lambda_j - \cE_j  \|_{\Diamond\bar{n}} \\
&\leq \bbE_{j \notin J} \| \Lambda_j - \cE_j  \|_{\Diamond\bar{n}} + {m \over N} \bbE_{j \in J} \| \Lambda_j - \cE_j  \|_{\Diamond\bar{n}} \label{22} \\
&\leq \sqrt{2 \ln(2) d^6 \log d \over m} + 4\sqrt{\bar{n} \over d} + {2m \over N} = \zeta,
\end{align}
where we have used that the diamond norm between two cptp maps is bounded by 2. Applying Markov's inequality gives
\begin{align}\label{Markov}
\mathbb P \left(  \| \Lambda_j - \cE_j  \|_{\Diamond\bar{n}} \geq {\zeta \over \delta} \right) \leq \delta,
\end{align}
which completes the proof.
\end{proof}

\subsubsection{Optimising over $m$ and $d $}

\noindent The parameters $d$ and $m$ can be freely chosen to minimize $\zeta$. We find
\beq
\zeta = f(d,m) \rightarrow f(d,m_{\text{min}}) = 4 \sqrt{\bar{n} \over d} + \left( {27 \ln(2) d^6 \log(d) \over N} \right)^{1/3}.
\eeq
In order to analytically optimise this over $d$, we can use the conservative bound $\ln d\leq d$:\\
\begin{align}
f(d,m_{\text{min}}) &= 4 \sqrt{\bar{n} \over d} + \left( {27 d^6 \ln(d) \over N} \right)^{1/3}\leq 4 \sqrt{\bar{n} \over d} + \left( {27 d^7  \over N} \right)^{1/3} = \tilde{f}(d,m_{\text{min}}). \label{f_d_preopt}
\end{align}
Optimising over $ d $ then gives
\beq
\tilde{f}(d,m_{\text{min}}) \rightarrow \tilde{f}(d_{\text{min}},m_{\text{min}}) = 17 \left(\frac{2}{7}\right)^{14/17}\left(\frac{\bar{n}^7}{N}\right)^{1/17} \approx 6.06 \left(\frac{\bar{n}^7}{N}\right)^{1/17}.
\eeq
Alternatively, we can numerically optimise equation (\ref{f_d_preopt}) directly. The results of this are shown in the main text, together with the analytical optimisation.



\section{Proving theorem two: Exponential cut-off \label{AppB}}
We begin by defining the \textit{modified Choi-Jamio\l kowski state} of a cptp map $\Lambda$ as $J(\Lambda) := \id_A \otimes \Lambda_{A'} [ \ket{\phi}\bra{\phi} ]$ for the entangled state $\ket{\phi}$ given by
\beq \label{phi}
\ket{\phi} := \cN \sum_{j=0}^{\infty} \phi_j \ket{j,j}_{AA'}
\eeq
where $\phi_j^2 \equiv 1 / e^{\omega j}$ and $\cN = \sqrt{1-e^{-\omega}}$, for $\omega >0$. Such a state is a two-mode squeezed state, with local reductions equal to thermal Gibbs  states
	$\gamma(\omega)$ given by
	\beq
	\gamma(\omega) = (1-e^{-\omega})e^{-\omega \hat n}.\label{gibbo}
	\eeq
	We note that such local thermal states have energy expectation $\tilde{n}\equiv\langle \hat{n} \rangle =  1/(e^\omega -1)$. We remark that for our specific choice of state, $(\Pi_d \otimes \bbI) \ket{\phi}=(\bbI\otimes \Pi_d) \ket{\phi}=\cN \sum_{j=0}^{d-1} \phi_j \ket{j,j}_{AA'}$.

\begin{lemma}\label{trunc_1}
For $L_{AB} = \tau - \sigma $, where $\tau=J(\Lambda_a)$ and $\sigma=J(\Lambda_b)$ are modified Choi-Jamio{\l}kowski states for cptp maps $\Lambda_a$ and $\Lambda_b$, we have
\begin{align}
\|L_{AB}\|_1 \leq d^2 \max_{\cM_B} \| \id_A \otimes \cM_B (L_{AB}) \|_1 + 4 \sqrt{ 1 \over e^{\omega d}},
\end{align}
where $d$ is a positive integer (the truncation dimension) and the maximum is taken over local measurement maps $\cM_B(Y) =
\sum_l \Tr(N_l Y) \ket{l}\bra{l}$, with a POVM $\{N_l\}$ and orthonormal states $\{\ket{l}\}$.
\end{lemma}

\noindent To prove Lemma \ref{trunc_1} we first need the following theorem.

\begin{theorem}
\label{truc_dist}
With $\rho=J(\Lambda)$ for cptp map $\Lambda$, and defining the truncated state $\rho_d:=(\Pi_d \otimes \bbI) \rho (\Pi_d \otimes \bbI)$ where $\Pi_{d} =\sum_{i=0}^{d-1} \proj{i}$, we have
\begin{equation}
\|\rho - \rho_d \|_1 \leq 2\sqrt{ 1 \over e^{\omega d}}.
\end{equation}
\end{theorem}

\begin{proof}
The gentle measurement Lemma in \cite{winter1999coding,ogawa2007making} shows that, for density matrix $\rho$ and linear operator $X$ s.t. $0 \leq X \leq \bbI$, we have
\beq \label{gent_meas}
\| \rho - \sqrt{X} \rho \sqrt{X} \|_1 \leq 2\sqrt{\Tr(\rho) - \Tr(\rho X)}.
\eeq
In our case we take $X=\sqrt{X} = \Pi_d \otimes \bbI$. Since $\Lambda$ is trace preserving, we have
\[
\Tr(\rho_d)=\Tr((\Pi_d \otimes \bbI) \proj{\phi} (\Pi_d \otimes \bbI))=\cN^2 \sum_{j=0}^{d-1} e^{-\omega j}=\cN^2 {1-e^{-\omega d} \over 1 - e^{-\omega} }=1-e^{-\omega d}.
\]
Since $\rho$ is normalized, $\Tr(\rho)=1$, the direct application of the gentle measurement lemma gives us the result.
\end{proof}

\noindent We can now prove Lemma \ref{trunc_1}.

\begin{proof}
(Lemma \ref{trunc_1}) Writing $L_{AB} = \sum_{ij =0}^{\infty} \ket{i}\bra{j} \otimes L_{ij}$ and $(\pi_A^d \otimes \id_B) (L_{AB}) =  (\Pi_d \otimes \bbI) L_{AB} (\Pi_d \otimes \bbI)$, we have
\begin{align}
\|L_{AB}\|_1	&\leq	\| (\pi_A^d \otimes \id_B) (L_{AB}) \|_1 +
				\| ((\id_A - \pi_A^d) \otimes \id_B) (L_{AB}) \|_1  \\
				&= \| (\pi_A^d \otimes \id_B) (L_{AB}) \|_1 +
				\|\sum_{i\vee j \geq d} \ket{i}\bra{j}\otimes L_{ij} \|_1  \\
&\leq d^2 \max_{\cM_B} \| (\pi_A^d \otimes \cM_B) (L_{AB}) \|_1 + \|\sum_{i\vee j \geq d} \ket{i}\bra{j}\otimes L_{ij} \|_1 \\
&\leq d^2 \max_{\cM_B} \| (\id_A \otimes \cM_B) (L_{AB}) \|_1 + \|\sum_{i\vee j \geq d} \ket{i}\bra{j}\otimes L_{ij} \|_1,
\end{align}
where where the first inequality comes from Lemma 1 in the Supplementary notes of \cite{Brandao2015QD}, and the second one comes from the Pinching theorem \cite{bhatia2013matrix}.

We have
\begin{align}
\|\sum_{i\vee j \geq d} \ket{i}\bra{j} &\otimes L_{ij} \|_1 \\
&= \|L_{AB} - (\Pi_{d} \otimes \bbI) L_{AB} (\Pi_{d} \otimes \bbI) \|_1 \nonumber \\
&= \|(\tau-\sigma) - (\Pi_{d} \otimes \bbI) (\tau-\sigma) (\Pi_{d} \otimes \bbI) \|_1 \nonumber \\
&= \|\tau - \tau_{d} - (\sigma - \sigma_{d} ) \|_1 \nonumber \\
&\leq \|\tau - \tau_{d} \|_1 + \| \sigma - \sigma_{d} \|_1 \nonumber \\
&< 4\sqrt{ 1 \over e^{\omega d}},\nonumber
\end{align}
where in the last line we used Theorem \ref{truc_dist}, and $\tau_{d}$ and $\sigma_d$ are the truncations of $\tau$ and $\sigma$, respectively.
\end{proof} ~\\


\begin{definition} (Exponential cut-off diamond norm) For a linear map $\Lambda $, and constants $\omega > 0$, $\Omega>1$, let
\beq 
\|\Lambda\|_{\Diamond\omega\Omega}:=
\sup_{
	\substack{
		\Tr\left[ \rho e^{\omega \hat{n}_{A'}} \right] \leq \Omega
	}
} \|\id_A \otimes \Lambda_{A'} [\rho]\|_1\,.
\eeq
\end{definition}

\begin{lemma}\label{exp_cut_lem2}
(Modification of Lemma 2 in \cite{Brandao2015QD} for infinite dimensional systems with an exponential cut-off) For cptp maps $ \Lambda_0 $ and $ \Lambda_1$ we have
\beq
\| \Lambda_0 - \Lambda_1 \|_{\Diamond \omega \Omega} \leq \tilde{d} \| J(\Lambda_0) - J(\Lambda_1) \|_1.
\eeq
where
\beq
\tilde{d} = {\Omega e^{\omega} \over e^{\omega} - 1}.
\eeq
and the \textit{modified Choi-Jamio\l kowski state} of a cptp map $\Lambda$ is $J(\Lambda) := \id_A \otimes \Lambda_{A'} [ \ket{\phi}\bra{\phi} ]$ for the entangled state $\ket{\phi}$ given in Eq.~(\ref{phi}).
\end{lemma}

\begin{proof}
Take $\ket{\psi}$ as the state optimal for the exponential cut-off diamond norm $\| \Lambda_0 - \Lambda_1 \|_{\diamond \omega \Omega}$ (there might not be an optimal state because the set is not closed, but we can get as close as we want). As noted in Remark~\ref{remark1} we only need to consider pure states. Any such state $\ket{\psi}=\sum_{i j=0}^{\infty} \psi_{ij} \ket{i,j}$ can written as $\ket{\psi}=(\Psi \otimes \bbI )\ket{\tilde{\psi}^+}$ for the (infinite) matrix of coefficients $\Psi=[\psi_{ij}]_{ij}$ and the unnormalized vector $\ket{\psi^+}=\sum_{j=0}^\infty \ket{j}\ket{j}$. It can also be obtained by local filtering on $\ket{\phi}$ as
$\ket{\psi} = (C \otimes \bbI )\ket{\phi}$
for $C_{ij}=\bra{i}C\ket{j}=\psi_{ij}/(\mathcal{N}\phi_j)=\psi_{ij}\sqrt{\frac{e^{\omega j}}{1-e^{-\omega}}}$. We can rewrite the matrix of coefficients $\Psi$ as $\Psi=\sqrt{1-e^{-\omega}}C(e^{\hat{n}\omega})^{-1/2}$, with the operator $\hat{n}$ diagonal in the considered basis. The reduced state $\rho_{A'}=\Tr_{A}(\proj{\psi}_{AA'})$ can be expressed as $\rho_{A'} = C^T C^* = (1-e^{-\omega})(e^{\hat{n}\omega})^{-1/2}C^TC^*(e^{\hat{n}\omega})^{-1/2}$. The exponential cut-off energy constraint on $\ket{\psi}$ can then be expressed as $\Omega \geq \Tr(\rho_{A'}e^{\hat{n}_{a'}\omega}) = (1-e^{-\omega}) \Tr((e^{\hat{n}_{A'}\omega})^{-1/2}C^TC^*(e^{\hat{n}_{A'}\omega})^{-1/2} e^{\hat{n}_{A'}\omega})=(1-e^{-\omega})\Tr(C^TC^*)$, where we have used the cyclic property of the trace. Thus, it must hold $\|C\|_\infty^2\leq \|C\|_2^2 =\Tr(C^\dagger C)\leq \Omega/(1-e^{-\omega})$.

We thus have
\[
\begin{aligned}
\| \Lambda_0 - \Lambda_1 \|_{\diamond \omega \Omega}
&= \|\id_A\otimes (\Lambda_0 - \Lambda_1)[\proj{\psi}]\|_1\\
&\leq \|C\|^2_\infty \|\id_A\otimes (\Lambda_0 - \Lambda_1)[\proj{\phi}]\|_1\\
&\leq \Omega/(1-e^{-\omega})\| J(\Lambda_0) - J(\Lambda_1) \|_1,
\end{aligned}
\]
which concludes the proof.
\end{proof}


\begin{lemma}\label{mutual}
For a density matrix $\rho_{AB} = \id_A \otimes \Lambda_{A'} [ \ket{\phi}\bra{\phi} ]$ where $\Lambda : \mathcal{D}(A') \rightarrow \mathcal{D}(B)$ is a cptp map and $\ket{\phi}$ is given in Eq.~(\ref{phi}), for $\rho_{AB}$ separable between $A$ and $B$ the mutual information is bounded by
\beq
I(A:B)_{\rho_{AB}} \leq {\varsigma}.
\eeq
Here ${\varsigma}$ is the entropy of the Gibbs state with energy $\tilde{n}=(e^{\omega}-1)^{-1}$, given by
\beq\label{bdd_en}
{\varsigma} = (\tilde{n}+1)\log (\tilde{n}+1) - \tilde{n} \log \tilde{n}.
\eeq
\end{lemma}

\begin{proof}
It was shown in \cite{horodecki1996quantum,nielsen2001separable} that, for any state that is separable between subsystems $A$ and $B$, we have $S(A),S(B) \leq S(A,B)$. The mutual information between $A$ and $B$ is given by $I(A:B) = S(A)+S(B)-S(A,B)$, and therefore $I(A:B) \leq S(A)$. Following \cite{winter2016tight}, the unique maximiser (among single-mode bosonic states) of the entropy subject to the constraint $\Tr(\rho \hat n) \leq \tilde n$ is the Gibbs state $\gamma(\omega)$ given in Eq.~\eqref{gibbo}. For such a state, \cite{winter2016tight,qi2016thermal} give the entropy as shown in Eq.~\eqref{bdd_en}. Since $\Lambda$ is trace preserving, the reduced state on $A$ of $\rho_{AB}$ is the same as that of $\ket{\phi}$.
\end{proof}


\begin{theorem}\label{main_th_exp}
Let $\Lambda  : \mathcal{D}(A) \rightarrow \mathcal{D}(B_1 \otimes \ldots \otimes B_N)$ be a cptp map. Define $\Lambda_j := \Tr_{ \backslash B_j} \circ \Lambda$ as the effective dynamics from $\mathcal{D}(A)$ to $\mathcal{D}(B_j)$ and fix a number $1 > \delta > 0$. Then there exists a POVM $\{M_k\}$ and a set $S \subseteq \{1,...,N\}$ with $|S| \geq (1-\delta)N$ such that, for all $j \in S$ (and for some truncation dimension $d$), we have
\beq
\|\Lambda_j - \cE_j  \|_{\diamond \omega \Omega} \leq {1 \over \delta} \zeta
\eeq
where the measure-and-prepare channel $\cE_j$ is given by
\beq
\cE_j(X) := \sum_l \Tr (M_k X) \sigma_{j,l},
\eeq
for states $\sigma_{j,l} \in \mathcal{D}(B_j)$, and
\beq
\zeta = \left( {27 \ln(2) \tilde{d}^2 d^4 {\varsigma} \over N} \right)^{1/3} + 4 \sqrt{ \tilde{d}^2 \over e^{\omega d}}
\eeq
where $\tilde{d} = {\Omega e^{\omega} / (e^{\omega} - 1)}$. Here ${\varsigma}$ is the entropy of the Gibbs state with energy $\tilde{n}$, as given above.
\end{theorem}

\noindent Before giving the proof we need a corollary of the proof of Theorem 1 in \cite{Brandao2015QD} (see Eq.~(16) in the Supplementary material of \cite{Brandao2015QD}).

\begin{corollary}\label{cor:brandao1}
For cptp map $\Lambda$ we define the modified Choi-Jamio\l kowski state of $\Lambda$ as $\rho_{AB_1...B_N}  := \id_{A} \otimes \Lambda_{A'} (\Phi_{AA'})$, where $\Phi = \ket{\phi}\bra{\phi}$ for $\ket{\phi} = \cN \sum \phi_k \ket{k,k}$ in Eq.~(\ref{phi}). The quantum-classical channels ${\cal M}_1,...,{\cal M}_N$ are defined as ${\cal M}_i(X) := \sum_l \Tr(N_{i,l}X) \ket{l}\bra{l}$, for POVM $\{N_{i,l}\}_l$. Fixing an integer $m$, there exists another integer $q\leq m$ and a set of indices $J:=(j_1,...,j_{q-1})$, such that
\beq
\bbE_{j \notin J} \max_{{\cal M}_j} \| \id_A \otimes {\cal M}_j \left( \rho_{AB_j} - \bbE_z \rho_A^z \otimes \rho_{B_j}^z \right) \|_1 \leq \sqrt{2 \ln(2) S(A) \over m}.
\eeq
Here the $j $'s are sampled from the uniform distribution over $\{0,1,...,N\}$ without replacement and $\rho_{AB_j}^z$ is the post-measurement state on $AB_j$ conditioned on obtaining $z=\{l_{j_1},...,l_{j_{q-1}}\}$ when measuring ${\cal M}_{j_1},...,{\cal M}_{j_{q-1}}$ in the subsystems $B_{j_1},...,B_{j_{q-1}}$ of $\rho$. The associated probability reads
\beq
p(z)=\Tr\left({\id_{A} \otimes \Lambda_{A'} (\Phi_{AA'})\,N_{j_1,l_{j_1}}\otimes...\otimes N_{j_{q-1},l_{j_{q-1}}}}\right)
\eeq
$S(A)$ is the entropy in subsystem $A$ of $\rho$.
\end{corollary}
\noindent Corollary \ref{cor:brandao1} can be obtained adapting the derivation in equations (9)--(16) in the supplementary material of \cite{Brandao2015QD},  with the entropy $\log d_A$ in \cite{Brandao2015QD} replaced by $S(A)$, and the definition of $\rho$ taken here as the \textit{modified} Choi-Jamio\l kowski state.

\begin{proof} (Theorem \ref{main_th_exp}) Using Lemma \ref{mutual}, take $S(A) = {\varsigma}$ in Corollary \ref{cor:brandao1}, which gives
\beq\label{corr1}
\bbE_{j \notin J} \max_{{\cal M}_j} \| \id_A \otimes {\cal M}_j \left( \rho_{AB_j} - \bbE_z \rho_A^z \otimes \rho_{B_j}^z \right) \|_1 \leq \sqrt{2 \ln(2) {\varsigma} \over m}.
\eeq
Lemma \ref{trunc_1} states that for truncation $d$ we have
\begin{align}
\|L_{AB}\|_1 \leq d^2 \max_{{\cal M}_B} \| \id_A \otimes {\cal M}_B (L_{AB}) \|_1 + 4 \sqrt{ 1 \over e^{\omega d}}.
\end{align}
We will see below that $\rho_{AB_j}$ and $\bbE_z \rho_A^z \otimes \rho_{B_j}^z$ are both modified Choi-Jamio\l kowski states, enabling the use of Lemma \ref{trunc_1} to give
\begin{align}
\|\rho_{AB_j} &- \bbE_z \rho_A^z \otimes \rho_{B_j}^z\|_1	 \\
&\leq d^2 \max_{{\cal M}_j} \| \id_A \otimes {\cal M}_j \left( \rho_{AB_j} - \bbE_z \rho_A^z \otimes \rho_{B_j}^z \right) \|_1 + 4 \sqrt{ 1 \over e^{\omega d}}. \notag
\end{align}
Combining this with Corollary \ref{cor:brandao1} we get
\begin{align}\label{18-l}
\bbE_{j \notin J}\|\rho_{AB_j} &- \bbE_z \rho_A^z \otimes \rho_{B_j}^z\|_1	 \\
&\leq d^2 \bbE_{j \notin J} \max_{{\cal M}_j} \| \id_A \otimes {\cal M}_j \left( \rho_{AB_j} - \bbE_z \rho_A^z \otimes \rho_{B_j}^z \right) \|_1 + 4 \sqrt{ 1 \over e^{\omega d}} \notag \\
&\leq d^2 \sqrt{2 \ln(2) {\varsigma} \over m} + 4 \sqrt{ 1 \over e^{\omega d}} \notag \\
&= \sqrt{2 \ln(2) d^4 {\varsigma} \over m} + 4 \sqrt{ 1 \over e^{\omega d}} \notag.
\end{align}
We now use Lemma \ref{exp_cut_lem2}, which states that
\beq
\| \Lambda_0 - \Lambda_1 \|_{\diamond \omega \Omega} \leq \tilde{d} \| J(\Lambda_0) - J(\Lambda_1) \|_1.
\eeq
where
\beq
\tilde{d} = {\Omega e^{\omega} \over e^{\omega} - 1}.
\eeq
Using this we find
\beq\label{20-l}
\| \Lambda_j - \cE_j  \|_{\diamond \omega \Omega} \leq \tilde{d} \| \rho_{AB_j} - \bbE_z \rho_A^z \otimes \rho_{B_j}^z \|_1.
\eeq
This is because the modified Choi-Jamio\l kowski state of $\Lambda_j$ is $\rho_{AB_j}$ by definition (see Corollary \ref{cor:brandao1} for the definition of the modified Choi-Jamio\l kowski state). The measure-and-prepare channel $\cE_j$ is explicitly given by
\beq\label{mp1}
\cE_j(X) := \cN^{-2} \bbE_z \Tr ((\rho_A^z)^T O X O^{\dagger}) \rho_{B_j}^z,
\eeq
where $O = \sum_i {1 \over \phi_i } \ket{i}\bra{i}$, and $\cN$ and $\phi_i$ are the same as in the entangled state $\Phi\equiv \ket{\phi}\bra{\phi}$ in Lemma \ref{exp_cut_lem2}. In Appendix~\ref{AppApp} we demonstrate that $\bbE_z \rho_A^z \otimes \rho_{B_j}^z$ is the modified Choi-Jamio\l kowski state of $\cE_j$.
Combining (\ref{18-l}) and (\ref{20-l}) gives
\begin{align}
\bbE_{j \notin J} \| \Lambda_j - \cE_j  \|_{\diamond \omega \Omega} &\leq \tilde{d} \bbE_{j \notin J} \| \rho_{AB_j} - \bbE_z \rho_A^z \otimes \rho_{B_j}^z \|_1 \\
&\leq \tilde{d} \left( \sqrt{2 \ln(2) d^4 {\varsigma} \over m} + 4 \sqrt{ 1 \over e^{\omega d}} \right) \\
&= \sqrt{2 \ln(2) \tilde{d}^2 d^4 {\varsigma} \over m} + 4 \sqrt{ \tilde{d}^2 \over e^{\omega d}}.
\end{align}
We then have
\begin{align}
\bbE_j \| \Lambda_j - \cE  \|_{\diamond \omega \Omega} &\leq \bbE_{j \notin J} \| \Lambda_j - \cE_j  \|_{\diamond \omega \Omega} + {m \over N} \bbE_{j \in J} \| \Lambda_j - \cE_j  \|_{\diamond \omega \Omega} \\
&\leq \sqrt{2 \ln(2) \tilde{d}^2 d^4 {\varsigma} \over m} + 4 \sqrt{ \tilde{d}^2 \over e^{\omega d}} + {2m \over N}.
\end{align}
We can minimise the right-hand-side with respect to $m$ to obtain
\begin{align}\label{exp_zeta}
\zeta = \left( {27 \ln(2) \tilde{d}^2 d^4 {\varsigma} \over N} \right)^{1/3} + 4 \sqrt{ \tilde{d}^2 \over e^{\omega d}}.
\end{align}
We use Markov's inequality (see Eq.~(\ref{Markov})) to complete the proof.
\end{proof}

\subsubsection{Minimising $\zeta $ with respect to the truncation dimension $d$.}

Our function to be minimised is of the form $(A/N^{1/3})d^{4/3}+B e^{-\omega d/2}$, where $A=\left( 27 \ln(2) \tilde{d}^2 {\varsigma} \right)^{1/3}$ and $B=4\tilde{d}$, see Eq.~\eqref{exp_zeta}. For $N\gg 1$, it can be shown that such a function is minimised by setting
\begin{align}
d\to d_{min}=\frac{2 W\left(\frac{81 B^3 N \omega ^4}{1024 A^3}\right)}{3 \omega }\simeq \frac{2 \ln \left(\frac{81 B^3 N \omega ^4}{1024 A^3}\right)}{3 \omega },
\end{align}
where $W(x)$ is the Lambert W function, which approximates to $\ln(x)$ for $x\gg1$. Substituting the above approximate value for $d_{\sf min}$ still gives a valid bound, since Eq.~\eqref{exp_zeta} holds for arbitrary $d$:
\beq
\zeta \leq
8 \left(\frac{\gamma_1}{N}\right)^{1/3}\left[1+\frac14\big(\ln(\gamma_2 N)\big)^{4/3}\right]\,,
\eeq
with
\begin{align}
\begin{split}
&\gamma_1 = \frac{2\tilde{d}^2s}{3\omega^4}\,, \quad \gamma_2=\frac{3 \tilde{d} \omega^4}{16s}\,, \\
&\tilde{n}=\displaystyle{1\over e^{\omega}-1}\,, \quad s= (\tilde{n}+1)\ln (\tilde{n}+1) - \tilde{n} \ln \tilde{n}\,.
\end{split}
\end{align}						
In the above, note that we have used $\ln(2)\,\varsigma=s,$ that is, we have simplified $\ln (2)\log(\cdot)=\ln(\cdot)$.

\section{Finding $\omega,\Omega$ for Gaussian states with bounded energy \label{AppC}}
Here we show that any subset of single-mode Gaussian states with bounded energy, i.e. ${\cal G}=\{\rho_G\vert \Tr[{\rho_G\hat n}]\leq \bar{n}\}$, obeys the desired cutoff condition for suitable values of $\omega,\Omega$. We recall the cutoff condition as
\begin{equation}\label{eq:cutoffgaussian}
\av{e^{\omega \hat n}}=\frac{2 \exp\left[{\frac{2\Re(\alpha)^2}{\kappa_+^2}+\frac{2\Im(\alpha)^2}{\kappa_-^2}}\right]}{(e^\omega-1)\kappa_+ \kappa_-} \leq \Omega\,,
\end{equation}
with $\kappa_\pm=\sqrt{\coth \left(\frac{\omega }{2}\right)-(2 m+1) e^{\pm 2 r}}$, and that the energy of a Gaussian state, in our chosen notation, reads $\langle\hat n\rangle=|\alpha|^2+m\cosh(2r)+\sinh^2r$. First, we note that $\kappa_-\geq\kappa_+$ (recall that we are assuming $r\geq0$), and by a simple constrained optimization, one can further show that $(2 m+1) e^{\pm 2 r}\leq3/2+2\bar n(2+\bar n)$. Hence, fixing $0<\epsilon<1$ and exploiting $\coth(\omega/2)>2/\omega$ for $\omega<2$, we find that $\kappa_+^2\geq\tfrac{2}{\omega}(1-\epsilon),$ provided that $\omega\leq\tfrac{2\epsilon}{3/2+2\bar n(2+\bar n)}$. Putting all this together, while recalling $e^\omega-1\geq \omega $ and $|\alpha|^2\leq \bar n$, we obtain
$\av{e^{\omega \hat n}}\leq \exp\left(\frac{\omega\bar n}{1-\epsilon}\right)/(1-\epsilon).$ Thus, for any $0<\epsilon<1$ we can satisfy our exponential cutoff condition for the set $\cal G$ by picking $\omega,\Omega$ such that
\begin{align}
\Omega&> 1/(1-\epsilon),\\
\omega&=\min\left\{ \frac{2\epsilon}{3/2+2\bar n(2+\bar n)}\, ,\,\frac{1-\epsilon}{\bar n}\ln\left((1-\epsilon)\Omega\right) \right\}.
\end{align}
Note that there is some freedom in the choice of the above parameters, which may be exploited for further optimization (e.g. by minimising $\gamma_1$ above).

\section{The measure-and-prepare channel $\cE_j$ in equation (\ref{mp1})\label{AppApp}}
\noindent The measure-and-prepare channel $\cE_j$ is explicitly given by
\beq
\cE_j(X) := \cN^{-2} \bbE_z \Tr ((\rho_A^z)^T O X O^{\dagger}) \rho_{B_j}^z,
\eeq
where $O = \sum_i {1 \over \phi_i } \ket{i}\bra{i}$, and $\cN$ and $\phi_i$ are the same as in the entangled state $\Phi\equiv \ket{\phi}\bra{\phi}$ in Lemma \ref{exp_cut_lem2}. We can confirm that $\bbE_z \rho_A^z \otimes \rho_{B_j}^z$ is the modified Choi-Jamio\l kowski state of $\cE_j$ as follows
\begin{align}
J(\cE_j) &= \id_A \otimes \cE_j(\Phi) \\
&= \cN^2 \sum_{kk'} \phi_k \phi_{k'} \id_A ( \ket{k}\bra{k'} ) \otimes \cE_j(\ket{k}\bra{k'}) \\
&= \sum_{kk'} \phi_k \phi_{k'} \ket{k}\bra{k'} \otimes \bbE_z \Tr ((\rho_A^z)^T O \ket{k}\bra{k'} O^{\dagger}) \rho_{B_j}^z \\
&= \sum_{kk'} \ket{k}\bra{k'} \otimes \bbE_z \Tr ((\rho_A^z)^T \ket{k}\bra{k'}) \rho_{B_j}^z \\
&= \bbE_z \sum_{kk'} \bra{k'} (\rho_A^z)^T \ket{k} \ket{k} \bra{k'} \otimes \rho_{B_j}^z \\
&= \bbE_z \rho_A^z \otimes \rho_{B_j}^z.
\end{align}
We still need to show that $\{M_z\} = \{ \cN^{-2} p(z) O^{\dagger} (\rho_A^z)^T O \}$ is a POVM. It is positive because for $X$ positive semidefinite, so is $O^{\dagger} X O$. We now show that completeness holds. First we have
\begin{align}
\rho_A &= \Tr_{\backslash A} (\bbI_A \otimes \Lambda[\Phi]) \\
&= \Tr_{\backslash A} (\bbI_A \otimes \Lambda \left[ \cN^2 \sum_{kk'} \phi_k \phi_{k'} \ket{kk}\bra{k'k'} \right] ) \\
&= \cN^2 \sum_{kk'} \phi_k \phi_{k'} \ket{k}\bra{k'} \Tr (\Lambda \left[ \ket{k}\bra{k'} \right] ) \\
&= \cN^2 \sum_{k} \phi_k^2 \ket{k}\bra{k} \left( = \rho_A^T \right)
\end{align}
Then
\begin{align}
\sum_z M_z &= \sum_z p(z) \cN^{-2} O^{\dagger} (\rho_A^z)^T O \\
&= \cN^{-2} O^{\dagger} (\sum_z p(z) \rho_A^z)^T O \\
&= \cN^{-2} O^{\dagger} (\rho_A)^T O \\
&= \cN^{-2} O^{\dagger} \rho_A O \\
&= \sum_j {1 \over \phi_j } \ket{j}\bra{j} \sum_{k} \phi_k^2 \ket{k}\bra{k} \sum_{j'} {1 \over \phi_{j'} } \ket{j'}\bra{j'} \\
&= \sum_{jj'k} {1 \over \phi_j } {1 \over \phi_{j'} }\phi_k^2 \ket{j}\bra{j} \ket{k}\bra{k} \ket{j'}\bra{j'} \\
&= \sum_{k} {1 \over \phi_k } {1 \over \phi_{k} } \phi_k^2 \ket{k}\bra{k} = \bbI,
\end{align}
where we have used $\rho_A = \sum_z p(z) \rho_A^z$ and $(\rho_A)^T= \rho_A$ due to it being diagonal in the computational basis.


\end{document}